\documentclass[11pt]{article}

\usepackage[margin=1in]{geometry}
\usepackage[T1]{fontenc}
\usepackage{lmodern}
\usepackage{amsmath,amssymb,amsthm}
\usepackage{bbm}
\usepackage{algorithm}
\usepackage[noend]{algpseudocode}
\usepackage{tikz}
\usepackage[textsize=footnotesize]{todonotes}
\usepackage[colorlinks=true,allcolors=blue]{hyperref}

\theoremstyle{plain}
\newtheorem{theorem}{Theorem}
\newtheorem{lemma}[theorem]{Lemma}

\theoremstyle{definition}
\newtheorem{definition}[theorem]{Definition}

  {\begin{quote}\textsc{#1.}\ \ignorespaces}%
  {\end{quote}}

\DeclareMathOperator{\dv}{div}

\DeclareMathOperator{\gr}{gr}
\newcommand{\E}{\mathbb{E}}
\newcommand{\ind}{\mathbbm{1}}
\renewcommand{\Pr}{\mathbb{P}}

\title{Rounding the Ball LP for Fair Max-Min Diversification}
\author{
  Juli\'an Mestre \qquad Lam Khai Trinh \qquad Anthony Wirth\\[4pt]
  \normalsize The University of Sydney}
\date{}

\begin{document}

\maketitle

\begin{abstract}
  Given~$n$ points in a metric space, partitioned into groups, $X_1,\dots,X_m$,
  and integer quotas, $k_1,\dots,k_m$, summing to~$k$, the \emph{Fair Max-Min
  Diversification} problem asks for a set of~$k$ points, exactly~$k_i$ from
  each group~$X_i$, maximizing the minimum pairwise distance. Addanki et
  al.~(ICDT 2022) described a ball LP for this problem and rounding algorithms
  that yield a factor 2 approximation whose fairness holds only in expectation
  and a factor 6 approximation with relaxed fairness guarantees, as well as an
  $(m+1)$-approximation with exact fairness.

  We introduce two new algorithms. The first method refines the rounding of Addanki
  et al., yielding a 2-approximate solution that is $\varepsilon$-fair with
  high probability, meaning that from every group $X_i$, at
  least~$(1-\varepsilon) k_i$ points are chosen.

  The second method in polynomial time returns a 4-approximation to the
  optimal value of the exact fairness version. In time $n^{O(1)} 2^{O(k)}$,
  which is fixed-parameter tractable in~$k$, we achieve an exactly fair
  4-approximation. Our method adapts the augmenting procedure behind Haxell's
  theorem~(Graphs Combin., 1995). This approximation factor does not depend
  on~$m$. Moreover, we show that no rounding of the ball LP achieves a
  smaller factor with exact fairness.

\end{abstract}

\section{Introduction}

Given a universe~$V$ of~$n$ elements and a metric (distance) function,~$d$,
the \emph{Max-Min diversification} problem asks for a~$k$-size subset~$S$
of~$V$ whose points are as spread out as possible. That is, we seek to
maximize the minimum distance between any two distinct points in~$S$. The
intent is to pick a representative sample of the universe, avoiding
redundancy. This max-min diversification is one of the most studied objectives
in data management.
Moumoulidou et al.~\cite{Moumoulidou0M21} introduced a \emph{fair} version of
this problem, where the universe~$V$ is partitioned into~$m$ disjoint groups
and the selected set must include a prescribed number,~$k_i$, of points from
group~$i$, for each~$i=1,\ldots,m$. As Moumoulidou et al.\ describe, ``the
goal is to maximize diversity in data selection with respect to numerical
attributes, while ensuring the satisfaction of fairness constraints with
respect to categorical ones''. For instance, a maps service summarizing the
restaurants in a city can pick a set that is geographically spread out while
including a prescribed number from each cuisine. These notions of fairness
capture, for example, proportional and equal representation. Fairness has
appeared as a criterion in algorithms for clustering and for
ranking~\cite{Moumoulidou0M21}.

Addanki et al.~\cite{Addanki0MM22} later introduced a linear-programming
relaxation of the problem, which we call the \emph{ball LP}, and showed that
rounding it gives substantially better approximation guarantees, both when the
fairness constraints only need to hold approximately and when they must hold
exactly.

We revisit the ball LP and give two new rounding schemes that push these
guarantees further. The first is a rounding algorithm that achieves diversity
factor~$2$ with approximate fairness, against the factor~$6$ guarantee of
Addanki et al.~\cite{Addanki0MM22}. The second is an algorithm for
\emph{exact} fairness, built by adapting the augmenting-path technique behind
Haxell's theorem on independent transversals. We discuss related work below
before summarizing our results.

\subsection{Related Work}

The unconstrained version of Max-Min diversification, where $m = 1$, is well
studied in facility location, information retrieval, and recommendation
systems (see Drosou and Pitoura~\cite{DrosouP10} for a survey of the area).
The simple greedy farthest-point algorithm, which Ravi et al.~\cite{RaviRT94}
call GMM, achieves a $2$-approximation, and no polynomial-time algorithm
improves on that factor unless $\mathrm{P} = \mathrm{NP}$.

Moumoulidou et al.~\cite{Moumoulidou0M21} were the first to study the fair
variant we consider here. They showed the problem stays NP-hard, that no
polynomial-time algorithm can beat a $2$-approximation unless $\mathrm{P} =
\mathrm{NP}$, and gave three algorithms: a $4$-approximation for the special
case of $m=2$ groups, a $(3m-1)$-approximation for general $m$, and a
$5$-approximation for constant $m$ and $k = o(\log n)$ that runs in time
exponential in $k$.

Addanki et al.~\cite{Addanki0MM22} introduced the ball LP that we build on
throughout this paper (Section~\ref{sec:prelim}). Rounding this LP, they
obtained a $2$-approximation whose fairness only holds in expectation, a
$6$-approximation with a $(1-\varepsilon)$-relaxed fairness guarantee that
holds with high probability, and an $(m+1)$-approximation with exact fairness,
already an improvement on the $(3m-1)$ bound of Moumoulidou et
al.~\cite{Moumoulidou0M21}. They also extended their approach to
constant-dimensional Euclidean space, where a $(1+\varepsilon)$-approximation
becomes achievable, and built streaming algorithms and composable coresets on
top of their rounding schemes.

Our exact-fairness algorithm uses the augmenting-path idea behind Haxell's
theorem on independent transversals~\cite{Haxell95a, Haxell01}, in the form
given algorithmically by Graf and Haxell~\cite{GrafH20}. Given a graph with
its vertices partitioned into groups, an independent transversal is an
independent set with exactly one vertex per group; Haxell's theorem gives a
combinatorial condition guaranteeing one exists. Asadpour et
al.~\cite{AsadpourFS12} use a similar augmenting idea for the Santa Claus
problem, grounding it in a configuration LP rather than a purely combinatorial
condition. We take a similar approach, but ground our augmenting procedure in
the ball LP. Specifically, our process relies on a bound on the LP weight of a
second (aka two-hop) neighborhood, which follows from the ball constraint and
the triangle inequality.

\subsection{Our Results}

We give two new rounding schemes for the ball LP of Addanki et
al.~\cite{Addanki0MM22}.

Section~\ref{sec:two} gives a greedy rounding that achieves diversity
factor~$2$ with $(1-\varepsilon)$-fairness, improving the factor~$6$ of
Addanki et al.\ for the same guarantee, under the hypothesis $k_i \ge
4\varepsilon^{-2}\ln(2m)$ for every $1 \le i \le m$. Our algorithm and their
algorithm share an analogous random process for ordering elements, based on a
solution to the ball LP, but have different rules for including elements in
the solution returned. Our method draws one point at a time and immediately
discards that point's entire ball, so what remains is the same process on a
smaller set. This self-similar structure supports a concentration bound on the
number of points selected from each group directly (Lemma~\ref{lem:tail}),
without relying on independence between selection events.

Section~\ref{sec:four} gives our main result, an iterative augmentation that
rounds the ball LP to a $4$-approximate \emph{exactly fair} solution, in time
$n^{O(1)}2^{O(k)}$, where $k = \sum_i k_i$ is the total quota. Unlike the
previous $(m+1)$-approximation, our factor of~$4$ does not degrade as the
number of groups,~$m$, grows, giving a strict improvement for $m \ge 4$. We
complement our positive result with an explicit instance showing that, with
exact fairness, the factor~$4$ is tight for any algorithm that rounds the ball
LP\@.

\section{Preliminaries}\label{sec:prelim}

An instance of the \emph{Fair Max-Min Diversification} problem is a finite
metric space, $(V,d)$, with $|V| = n$, a partition of $V$ into groups
$X_1,\dots,X_m$, and integer quotas $k_i \ge 1$ with $|X_i| \ge k_i$, for
every $1 \le i \le m$. We write $k = \sum_i k_i$ for the total quota, and
assume $k \ge 2$. The objective is to find a set~$S \subseteq V$
with maximum \emph{diversity},
\[\dv(S) = \min\{d(u,v) : u,v \in S,\ u \ne v\}\,,\]
and \emph{exact fairness},
\[ |S \cap X_i| = k_i\,, \qquad \forall 1 \le i \le m\,.\]
We denote the optimal value of the problem by
\[
  \ell^{*} = \max\{\dv(S) : S \text{ exactly fair}\}\,.
\]
Section~\ref{sec:two} relaxes the fairness requirement: For
$\varepsilon \in (0,1)$, set~$S$ is \emph{$\varepsilon$-fair} whenever
\[
  |S \cap X_i| \ge (1-\varepsilon)k_i\,,
  \qquad \forall 1 \le i \le m\,.
\]

\subsection{The conflict graph and the ball LP}

Following the treatment of Addanki et al.~\cite{Addanki0MM22}, for a given
guess $\gamma > 0$ of the optimal value $\ell^{*}$, we define the auxiliary
structures: the conflict graph and the ball LP.

\begin{definition}\label{def:conflict}
For a threshold $\gamma > 0$, the \emph{conflict graph} $H_\gamma$ is the
graph on $V$ in which distinct $u, v \in V$ are adjacent exactly when
$d(u,v) < \gamma$.
\end{definition}
Note that a subset $S \subseteq V$ is independent in $H_\gamma$ if and only if
$\dv(S) \ge \gamma$.

Throughout, for a vector $x \in \mathbb{R}^V$ and a subset $W \subseteq V$,
we write $x(W) = \sum_{u \in W} x(u)$. Also,
$B(p,r) = \{u \in V : d(p,u) < r\}$ is the \emph{open} ball of radius~$r$
about $p \in V$.

\begin{definition}\label{def:balllp}
For a threshold $\gamma > 0$, the \emph{ball LP at $\gamma$} is
\[
  \mathcal{P}_\gamma = \Bigl\{\, x \in [0,1]^V :
  x(X_i) \ge k_i \quad \forall 1 \le i \le m, \quad
  x\bigl(B(p,\gamma/2)\bigr) \le 1 \quad \forall p \in V \,\Bigr\}\,,
\]
whose two constraint families are the \emph{fairness} and the \emph{ball}
constraints, respectively. We call~$\gamma$ \emph{feasible} if
$\mathcal{P}_\gamma \ne \varnothing$, and let
$\gamma_{\max} = \sup\{\gamma : \gamma \text{ feasible}\}$.
\end{definition}
Addanki et al.~\cite{Addanki0MM22} call this the Fair Max-Min LP\@. They showed
that every $\gamma \le \ell^{*}$ is feasible, so
\begin{equation}\label{eq:gammamax}
  \gamma_{\max} \ge \ell^{*}\,.
\end{equation}
As~$\gamma$ increases, the balls~$B(p,\gamma/2)$ change only as~$\gamma$
passes $2\,d(u,v)$ for some pair $u,v \in V$ with $u \ne v$.
Consequently,~$\gamma_{\max}$ is one of these $O(n^2)$ values and is itself
feasible.
Solving the ball LP at each of these thresholds yields~$\gamma_{\max}$
together with some~$x \in \mathcal{P}_{\gamma_{\max}}$, in polynomial time.

\section{2-approximation with relaxed fairness}\label{sec:two}

In this section, we present a randomized algorithm whose output is always
$2$-approximate in diversity and is $\varepsilon$-fair with constant
probability, provided every quota is at least $4\varepsilon^{-2}\ln(2m)$. The
success probability can be boosted arbitrarily close to~$1$ by running the
algorithm until its output is $\varepsilon$-fair.

\begin{theorem}\label{thm:two}
Let $\varepsilon \in (0,1)$ and suppose $k_i \ge 4\varepsilon^{-2}\ln(2m)$ for
every $1 \le i \le m$.
There is a randomized polynomial-time algorithm that outputs a set
$S \subseteq V$ with $\dv(S) \ge \ell^{*}/2$ that, with probability at
least~$1/2$,
satisfies
\[
  |S \cap X_i| \ge (1-\varepsilon)k_i\,, \qquad \forall 1 \le i \le m\,.
\]
\end{theorem}

We note that Addanki et al.~\cite[Thm.~7]{Addanki0MM22} proved a similar
guarantee under the hypothesis $k_i \ge 3\varepsilon^{-2}\ln(2m)$ for every
$1 \le i \le m$, but with diversity factor~$6$, in place of~$2$. For small
quotas, that is for $k = O(\log n)$, the algorithm of Theorem~\ref{thm:four}
in Section~\ref{sec:four} achieves diversity factor~$4$ with exact fairness in
polynomial time.

Factor~$2$ cannot be improved, no matter how far we relax the fairness requirement: Addanki et al.~\cite[Thm.~10]{Addanki0MM22} prove that for all
constants $\alpha < 2$ and $\beta > 0$, unless $\mathrm{P} = \mathrm{NP}$, no
polynomial-time algorithm achieves diversity factor~$\alpha$ and at
least~$\beta k_i$ points per group.

\subsection{Greedy rounding}

\emph{Greedy rounding} (Algorithm~\ref{alg:round}) takes a feasible
threshold~$\gamma$ and a vector
$x \in \mathcal{P}_\gamma$, and draws an independent set of the conflict
graph~$H_{\gamma/2}$. We write~$N(v)$ for the neighborhood of $v$ and
$N[v] = N(v) \cup \{v\}$ for the closed neighborhood; in this section, all
neighborhoods are taken with respect to $H_{\gamma/2}$.

\begin{algorithm}[ht]
\caption{Greedy rounding of $x \in \mathcal{P}_\gamma$.}\label{alg:round}
\begin{algorithmic}[1]
\State $R \gets V$, $S \gets \varnothing$
\While{$x(R) > 0$}
  \State draw $v \in R$ with probability $x(v)/x(R)$
  \State $S \gets S \cup \{v\}$, $R \gets R \setminus N[v]$
\EndWhile
\State \Return $S$
\end{algorithmic}
\end{algorithm}

The output is an independent set $S$ of $H_{\gamma/2}$: Each chosen vertex
removes its entire closed neighborhood, so no later choice is adjacent to an
earlier one. Moreover, each iteration removes at least one vertex from~$R$, so
the algorithm runs in polynomial time.

For $R \subseteq V$ write $S(R)$ for the output of the loop of
Algorithm~\ref{alg:round} started with that $R$.
The process is self-similar: conditioned on the first draw $v$, the rest of
the run is the process started from $R \setminus N[v]$, so
$S(R) = \{v\} \cup S(R \setminus N[v])$. The analysis rests on one consequence
of the ball constraints. A point lies at distance less than $\gamma/2$ from
$v$ exactly when it is $v$ itself or is adjacent to $v$, so
$N[v] = B(v,\gamma/2)$ for every $v \in V$ and therefore
\begin{equation}\label{eq:nbhd}
  x\bigl(N[v]\bigr) \le 1\,, \qquad \forall v \in V\,.
\end{equation}
This bound survives restriction to any~$R$:
$x\bigl(N[v] \cap R\bigr) \le x\bigl(N[v]\bigr)$. The analysis below uses only
the self-similarity and the bound~\eqref{eq:nbhd}.

Addanki et al.~\cite{Addanki0MM22} round a vector in $\mathcal{P}_\gamma$
using a different algorithm: They draw a random ordering of the points in the
support of $x$, sampling without replacement with probabilities proportional
to $x$. They keep every point that comes first in its own ball of radius
$\gamma/2$. The output meets the quotas \emph{only in expectation}.

Under their rule a point is discarded as soon as some point of its ball
precedes it. This holds whether that point was itself kept or discarded. A
discarded point therefore goes on to discard others. Suppose $H_{\gamma/2}$ is
the path $(u, v, w)$ and their algorithm draws the ordering~$u, v, w$. Their
rule keeps only $u$: It discards $v$ because $u$ precedes it, and it discards
$w$ because the discarded $v$ precedes it. On the other hand, Algorithm~\ref{alg:round},
conditional on the same first draw of~$u$, returns~$\{u, w\}$.
Indeed, the two
executions can be coupled so that our output always contains
theirs.\footnote{Give each point~$v$ an independent exponential clock of rate~$x(v)$. We let both algorithms follow the same clock order. A draw from~$R$
with probability~$x(v)/x(R)$ is the same as taking the first element of~$R$ in
the clock order; since the exponential clocks are memoryless, this holds at
every step. Because $B(v,\gamma/2) = N[v]$, their rule keeps exactly those~$v$
with the smallest clock in~$N[v]$. Greedy rounding puts every such~$v$ into~$S$: it drops~$v$ from~$R$ only when it takes a member of~$N[v]$, and every
other member of~$N[v]$ has a larger clock, so none is taken before~$v$.}
In deciding which items to include, our approach echoes some of the {\em Pivot} algorithm of correlation clustering ~\cite{10.1145/1411509.1411513}.

\subsection{The concentration inequality}

In this subsection we show that for every set $A \subseteq V$, the size
$|S \cap A|$ is unlikely to fall far below $x(A)$. Later, in the proof of
Theorem~\ref{thm:two}, we apply the bound to show that the output of greedy
rounding is likely to be $\varepsilon$-fair. The events $v \in S$ are not
independent, since whether $v$ enters $S$ depends on every earlier draw, so we
cannot use Chernoff to get the desired tail bound. Instead, we bound
$\E\bigl[e^{-\lambda |S \cap A|}\bigr]$ directly in Lemma~\ref{lem:mgf}, by
induction on the remaining set~$R$, using only the self-similarity of the
random process and~\eqref{eq:nbhd}. Markov's inequality applied to
$e^{-\lambda |S \cap A|}$ then turns this into a tail bound for $|S \cap A|$.
Both bounds are in terms of~$x(A)$, not $\E\bigl[|S \cap A|\bigr]$; this is
enough for us, as the fairness constraint guarantees~$x(X_i) \ge k_i$.

\begin{lemma}\label{lem:mgf}
Let~$\gamma$ be feasible, let $x \in \mathcal{P}_\gamma$, let~$S$ be the
output of greedy rounding (Algorithm~\ref{alg:round}),
and let $A \subseteq V$. Then for every $\lambda \ge 0$,
\begin{equation*}
  \E\Bigl[e^{-\lambda |S \cap A|}\Bigr] \le b^{-x(A)}\,,
\end{equation*}
where $b = 2 - e^{-\lambda} \in [1,2)$.
\end{lemma}

\begin{proof}
For $R \subseteq V$ let $a(R) = x(A \cap R)$ and
$\Phi(R) = \E\bigl[e^{-\lambda |S(R) \cap A|}\bigr]$. We prove the following
by induction on~$|R|$:
\begin{equation}
  \Phi(R) \le b^{-a(R)}\,. \label{eq:IH}
\end{equation}
The base case is $x(R) = 0$: the process stops at once, so
$S(R) = \varnothing$, and $a(R) = 0$ since $x$ is non-negative, so both sides
of~\eqref{eq:IH} are~$1$. For the inductive case, suppose $x(R) > 0$. For
$v \in R$ let $c(v,R) = x(A \cap N[v] \cap R)$. We observe that

\begin{enumerate}
\item $c(v,R) \le x(N[v]) \le 1$,
\item $x\bigl(A \cap (R \setminus N[v])\bigr) = a(R) - c(v,R)$, and
\item $\sum_{v \in R} x(v)\,c(v,R) \le a(R)$.
\end{enumerate}

The first observation follows from~\eqref{eq:nbhd}; the second from the
definitions of~$a(R)$ and~$c(v,R)$; the third from neighborhoods being mutual,
namely,
  \[
    \sum_{v \in R} x(v)\, c(v,R)
    = \sum_{v \in R}\, \sum_{u \in A \cap N[v] \cap R} x(v)\,x(u)
    = \sum_{u \in A \cap R} x(u)\, x\bigl(N[u] \cap R\bigr)
    \le \sum_{u \in A \cap R} x(u)
    = a(R)\,.
  \]

Write $\ind[v \in A]$ for the indicator of $v \in A$. Conditioning on the
first draw in Algorithm~\ref{alg:round},
invoking self-similarity for the equality, and applying
inequality~\eqref{eq:IH} to $R \setminus N[v]$, inductively, via the second
observation, we get
\begin{equation}
  \Phi(R)
  = \sum_{v \in R} \frac{x(v)}{x(R)}\, e^{-\lambda \ind[v \in A]}\,
        \Phi\bigl(R \setminus N[v]\bigr)
  \le \frac{b^{-a(R)}}{x(R)} \sum_{v \in R} x(v)\,
        e^{-\lambda \ind[v \in A]}\, b^{c(v,R)}\,.
  \label{eqn:phir}
\end{equation}
The function $t \mapsto b^{t}$ is convex, so for $c(v,R) \in [0,1]$ it follows
that $b^{c(v,R)} \le 1 + (b-1)c(v,R)$. Hence
\begin{align*}
  \sum_{v \in R} x(v)\, e^{-\lambda \ind[v \in A]}\, b^{c(v,R)}
  &\le \sum_{v \in R} x(v)\, e^{-\lambda \ind[v \in A]}
     + (b-1) \sum_{v \in R} x(v)\, e^{-\lambda \ind[v \in A]}\, c(v,R) \\
  \intertext{The first sum here is $x(R\setminus A) + e^{-\lambda}x(A\cap R) = x(R) - (1-e^{-\lambda})a(R)$,
  and the second sum is bounded using $e^{-\lambda\ind[v\in A]}\le1$, so}
  &\le \bigl(x(R) - (1 - e^{-\lambda})\,a(R)\bigr)
     + (b-1) \sum_{v \in R} x(v)\, c(v,R)\,, \\
  \intertext{which, by the third observation and $b-1=1-e^{-\lambda}$, gives}
  &\le \bigl(x(R) - (1 - e^{-\lambda})\,a(R)\bigr) + (b-1)\,a(R)
   = x(R)\,.
\end{align*}
Substituting this bound into the right-hand side
of inequality~\eqref{eqn:phir}, we have $\Phi(R) \le b^{-a(R)}$. Hence, by
induction,~\eqref{eq:IH} holds for all $R \subseteq V$. Taking $R = V$ yields
the lemma.
\end{proof}

\begin{lemma}\label{lem:tail}
In the setting of Lemma~\ref{lem:mgf}, for every $\varepsilon \in (0,1)$,
\[
  \Pr\bigl[\, |S \cap A| \le (1-\varepsilon)x(A) \,\bigr]
  \le e^{-\varepsilon^{2}x(A)/4}\,.
\]
\end{lemma}

\begin{proof}
For any $\lambda > 0$, let $b = 2 - e^{-\lambda}$ be as in Lemma~\ref{lem:mgf}
and $t = b - 1 = 1 - e^{-\lambda}$. Note that $\lambda = -\ln(1-t)$ and~$t$
ranges over $(0,1)$. Markov's inequality applied to the positive variable
$e^{-\lambda |S \cap A|}$, followed by Lemma~\ref{lem:mgf}, gives
\begin{align}
  \Pr\bigl[\, |S \cap A| \le (1-\varepsilon)x(A) \,\bigr]
  &= \Pr\bigl[e^{-\lambda |S \cap A|}
       \ge e^{-\lambda(1-\varepsilon)x(A)}\bigr] \nonumber \\
  &\le e^{\lambda(1-\varepsilon)x(A)}\,
       \E\bigl[e^{-\lambda |S \cap A|}\bigr] \nonumber \\
  &\le e^{\phi(t)\,x(A)}\,,
  \label{eqn:markov}
\end{align}
where
\begin{equation}
  \phi(t)
  = \lambda(1-\varepsilon) - \ln b
  = -(1-\varepsilon)\ln(1-t) - \ln(1+t)
  = -\ln(1-t^{2}) + \varepsilon\ln(1-t)\,.
  \label{eqn:phit}
\end{equation}
Inequality~\eqref{eqn:markov} holds for every $t \in (0,1)$. To prove the
lemma it therefore suffices to find one~$t$ with
$\phi(t) \le -\varepsilon^{2}/4$. As a warm-up, we first exhibit one with
$\phi(t) \le -\varepsilon^{2}/8$.

We use the fact that $\ln(1-z) \le -z$ for every $z < 1$. At $z = t$ this
yields $\ln(1-t) \le -t$, and at $z = -\frac{t^{2}}{1-t^{2}}$ this yields
$-\ln(1-t^{2}) \le \frac{t^{2}}{1-t^{2}}$. Hence, inequality~\eqref{eqn:phit} gives
\[
  \phi(t) \le \frac{t^{2}}{1-t^{2}} - \varepsilon t\,.
\]
Evaluating the above at $t = \varepsilon/4$ we have $t^{2} \le 1/16$, so
$\frac{t^{2}}{1-t^{2}} \le 2t^{2}$, and therefore
\[
  \phi(\varepsilon/4)
  \le \frac{\varepsilon^{2}}{8} - \frac{\varepsilon^{2}}{4}
  = -\frac{\varepsilon^{2}}{8}\,.
\]
Combined with~\eqref{eqn:markov}, this implies
\[
  \Pr\bigl[\, |S \cap A| \le (1-\varepsilon)x(A) \,\bigr]
  \le e^{-\varepsilon^{2}x(A)/8}\,.
\]

In order to get the sharper bound in the lemma statement, we use the complete
Taylor series
$\ln(1-z) = -\sum_{n \ge 1} \frac{z^{n}}{n}$, which holds for $|z| < 1$.
Using the first two terms gives
\[\ln(1-z) \le -z - \frac{z^{2}}{2}\,,\]
which at $z = t$ yields $\ln(1-t) \le -t - \frac{t^{2}}{2}$.
Meanwhile,
bounding the Taylor series terms with $n \ge 2$ by $\tfrac12 \sum_{n \ge 2} z^{n}$ gives
\[\ln(1-z) \ge -z - \frac{z^{2}}{2(1-z)}\,.\]
At $z = t^{2}$, this
yields $-\ln(1-t^{2}) \le t^{2} + \frac{t^{4}}{2(1-t^{2})}$. Hence
inequality~\eqref{eqn:phit} gives
\[
  \phi(t) \le t^{2} + \frac{t^{4}}{2(1-t^{2})}
    - \varepsilon t - \frac{\varepsilon t^{2}}{2}\,.
\]
Evaluating the above at~$t = \varepsilon/2$, we have~$t^{2} \le 1/4$, so
$\frac{t^{4}}{2(1-t^{2})} \le t^{4}$, and therefore
\[
  \phi(\varepsilon/2)
  \le \frac{\varepsilon^{2}}{4} + \frac{\varepsilon^{4}}{16}
      - \frac{\varepsilon^{2}}{2} - \frac{\varepsilon^{3}}{8}
  = -\frac{\varepsilon^{2}}{4} - \frac{\varepsilon^{3}}{8}
      + \frac{\varepsilon^{4}}{16}
  \le -\frac{\varepsilon^{2}}{4}\,,
\]
the last inequality holds because $\varepsilon < 1$. Combined
with~\eqref{eqn:markov}, this implies
\[
  \Pr\bigl[\, |S \cap A| \le (1-\varepsilon)x(A) \,\bigr]
  \le e^{-\varepsilon^{2}x(A)/4}\,,
\]
and the lemma follows.
\end{proof}

\subsection{Proof of Theorem~\ref{thm:two}}

\begin{proof}
Compute $\gamma_{\max}$ and a vector $x^{*} \in \mathcal{P}_{\gamma_{\max}}$
in polynomial time (see Section~\ref{sec:prelim} for details).
Run greedy rounding (Algorithm~\ref{alg:round}) on~$x^{*}$ and let~$S$ be its
output.

\emph{Diversity.} $S$ is independent in $H_{\gamma_{\max}/2}$, so
$\dv(S) \ge \gamma_{\max}/2$, which is at least $\ell^{*}/2$
by~\eqref{eq:gammamax}.
We note that this always holds, regardless of the random choices the algorithm
makes.

\emph{Fairness.} Fix a group~$X_i$. The fairness constraint of
$\mathcal{P}_{\gamma_{\max}}$ gives $x^{*}(X_i) \ge k_i$, so
$(1-\varepsilon)k_i \le (1-\varepsilon)x^{*}(X_i)$ and Lemma~\ref{lem:tail}
with $A = X_i$ gives
\[
  \Pr\bigl[\,|S \cap X_i| \le (1-\varepsilon)k_i\,\bigr]
  \le e^{-\varepsilon^{2}x^{*}(X_i)/4}
  \le e^{-\varepsilon^{2}k_i/4}
  \le \frac{1}{2m}\,.
\]
The second step uses $x^{*}(X_i) \ge k_i$, and the last uses the assumption of
Theorem~\ref{thm:two} that $k_i \ge 4\varepsilon^{-2}\ln(2m)$.
Since $X_i$ was arbitrary, a union bound over the~$m$ groups gives
\[
  \Pr\bigl[\, S \text{ is } \varepsilon\text{-fair} \,\bigr]
  \ge 1 - \sum_{i=1}^{m}
      \Pr\bigl[\, |S \cap X_i| \le (1-\varepsilon)k_i \,\bigr]
  \ge 1 - m \cdot \frac{1}{2m}
  = \frac{1}{2}\,.
\]
\end{proof}

The output can also be capped at the quotas. Delete points from any group~$X_i$ with $|S \cap X_i| > k_i$ until $|S \cap X_i| = k_i$. Deleting points
cannot lower a minimum distance over pairs, so $\dv(S) \ge \ell^{*}/2$
survives. The deletions stop at~$k_i$, which is at least~$(1-\varepsilon)k_i$,
so the lower bounds survive as well.

Finally, we note that even if the groups overlap, the algorithm works in the
following sense: given a vector $x \in \mathcal{P}_\gamma$, greedy rounding
yields a set~$S$ with $\dv(S) \ge \gamma/2$ that is $\varepsilon$-fair with
probability at least~$\tfrac12$, provided $k_i \ge 4\varepsilon^{-2}\ln(2m)$
for every $1 \le i \le m$. Indeed, the argument never uses the fact that the
groups are disjoint.
Capping the output at the quotas, however, does use disjointness: with
overlapping groups a single deletion can reduce two counts at once.

\section{4-approximately diverse with exact fairness}\label{sec:four}

In this section, we estimate in polynomial time, within a factor
of~$4$, the value~$\ell^{*}$; in time exponential in~$k$, we produce an exactly fair set of
diversity at least~$\ell^{*}/4$. The factor doubles that of
Section~\ref{sec:two} because the rounding needs the total $x$-value of every
second-neighborhood in the conflict graph to be at most~$1$, which forces the
threshold down to $\gamma/4$.

We follow Feige~\cite{Feige08} in distinguishing an \emph{approximation}
algorithm, which must return a feasible \emph{solution} whose objective value
is close to the optimum, from an \emph{estimation} algorithm, which need only
return a \emph{number} close to the optimum objective function value.

\begin{theorem}\label{thm:four}
Fair Max-Min Diversification admits a polynomial-time $4$-estimation algorithm
and a $n^{O(1)} 2^{O(k)}$-time $4$-approximation algorithm. The former returns
a number $\gamma$ with $\ell^{*} \le \gamma \le 4\,\ell^{*}$; the latter
returns an exactly fair set,~$S$, with $\dv(S) \ge \ell^{*}/4$.
\end{theorem}

The algorithm rounds $x \in \mathcal{P}_\gamma$ to an independent set of
the conflict graph~$H_{\gamma/4}$. For a vertex $v$ of a graph, we use
\(
  N^2[v] = \bigcup_{u \in N[v]} N[u]
\)
to denote the closed second-neighborhood of~$v$, that is, the set of vertices
at ``distance'' at most~$2$ from~$v$ in~$H_{\gamma/4}$.

Lemma~\ref{lem:round} is the rounding step. With this lemma, we immediately
deduce Theorem~\ref{thm:four}. After we introduce the tight instance, the remainder of this section comprises a proof of Lemma~\ref{lem:round}.

\begin{lemma}\label{lem:round}
Let $\gamma > 0$ and let $x \in \mathcal{P}_\gamma$. Then $H_{\gamma/4}$ has
an independent set~$S$ with $|S \cap X_i| = k_i$ for every~$i$, and such an
$S$ can be computed in $n^{O(1)} 2^{O(k)}$ time.
\end{lemma}

\begin{proof}[Proof of Theorem~\ref{thm:four}]
Both algorithms begin by computing~$\gamma_{\max}$ together with a vector
$x \in \mathcal{P}_{\gamma_{\max}}$, which takes polynomial time
(Section~\ref{sec:prelim}). We have $\gamma_{\max} \ge \ell^{*}$
by~\eqref{eq:gammamax}. In the other direction, Lemma~\ref{lem:round} applied
to~$x$ produces an independent set~$S$ of~$H_{\gamma_{\max}/4}$ that meets
every quota exactly, so~$S$ is exactly fair and its points are pairwise at
distance at least~$\gamma_{\max}/4$, whence
$\ell^{*} \ge \dv(S) \ge \gamma_{\max}/4$. The estimation algorithm returns~$\gamma_{\max}$, and the approximation algorithm returns~$S$, which
Lemma~\ref{lem:round} computes in $n^{O(1)} 2^{O(k)}$ time.
\end{proof}

\subsection{Tight instance}

In this subsection we show that no rounding algorithm for the ball LP that
outputs an exactly fair solution can improve the factor of~$4$ in
Theorem~\ref{thm:four}.
We construct an instance with $\ell^{*} = 1$ on
which~$\mathcal{P}_\gamma$ is feasible at~$\gamma = 4$.

\begin{lemma}\label{lem:witness}
There is an instance with $n=12$, $m=4$ and $k_i = 1$ for all $i$, for which
$\ell^{*} = 1$ while $\mathcal{P}_\gamma$ is feasible at~$\gamma = 4$.
\end{lemma}

\begin{proof}
Let $V$ be the union of three four-point \emph{clusters}
\[
  C_t = \{a_t,a'_t,b_t,b'_t\}, \qquad 1 \le t \le 3,
\]
each carrying a copy of $K_{2,2}$ with \emph{shores} $\{a_t,a'_t\}$ and
$\{b_t,b'_t\}$. Let $d$ be the shortest-path distance in this graph with unit
edge lengths, capped at~$4$: two points are at distance~$1$ if they lie on
opposite shores of a cluster, $2$ if they lie on the same shore, and~$4$ if
they lie in different clusters. We partition $V$ into four groups
\[
\begin{array}{ll}
  X_1 = \{b_1, a_2, a'_2\}, & X_2 = \{a_1, a_3, a'_3\},\\[2pt]
  X_3 = \{a'_1, b_3, b'_3\}, & X_4 = \{b'_1, b_2, b'_2\},
\end{array}
\]
and set every quota,~$k_i$, to~$1$.
Observe that each group holds one point of $C_1$ and one full shore of $C_2$ or of $C_3$:
the two shores of $C_2$ go to $X_1$ and $X_4$, and those of $C_3$ go to
$X_2$ and $X_3$. Figure~\ref{fig:witness} shows the instance.

\begin{figure}[t]
\centering
\begin{tikzpicture}[
  x=1.5cm, y=1.5cm,
  pt/.style={circle, draw, semithick, inner sep=0pt, minimum size=4mm,
             font=\scriptsize},
  label distance=2.5pt,
  every label/.append style={font=\footnotesize, inner sep=1.5pt},
]
% #1 x offset, #2 cluster index, then the groups of a_t, a'_t, b_t, b'_t
\newcommand{\clu}[6]{%
  \begin{scope}[shift={(#1,0)}]
    \node[pt, label={above left:$a_{#2}$}]   (a#2) at (0,1) {#3};
    \node[pt, label={below left:$a'_{#2}$}]  (c#2) at (0,0) {#4};
    \node[pt, label={above right:$b_{#2}$}]  (b#2) at (1,1) {#5};
    \node[pt, label={below right:$b'_{#2}$}] (d#2) at (1,0) {#6};
    \draw[semithick] (a#2) -- (b#2) (a#2) -- (d#2)
                     (c#2) -- (b#2) (c#2) -- (d#2);
    \node[font=\small] at (0.5,-0.8) {$C_{#2}$};
  \end{scope}%
}
\clu{0}{1}{2}{3}{1}{4}
\clu{2.5}{2}{1}{1}{4}{4}
\clu{5}{3}{2}{2}{3}{3}
\end{tikzpicture}
\caption{The instance of Lemma~\ref{lem:witness}. Each cluster is a copy of
$K_{2,2}$, with its two shores drawn as the left and right columns. Adjacent
points are at distance~$1$, the two points of a shore at~$2$, and points in
different clusters at~$4$. The number inside a vertex is the index of the
group holding it: each shore of $C_2$ and of $C_3$ lies in a single group,
while $C_1$ has one point in each of the four groups.}%
\label{fig:witness}
\end{figure}
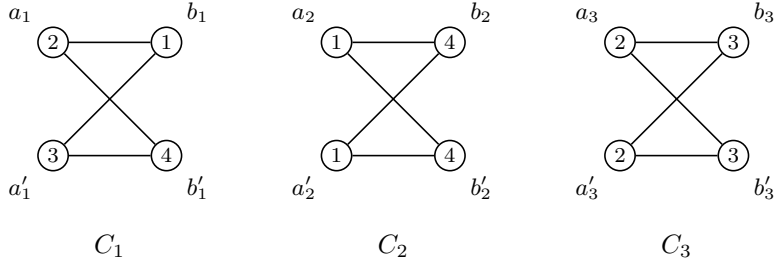

\emph{The ball LP is feasible at $\gamma = 4$.} Take the uniform vector $x$
with $x(u) = \tfrac13$ for all $u$. For each group we have $x(X_i) = 1 = k_i$.
Since the ball is open, $B(p,2)$ holds only $p$ and the two points on the
opposite shore of $p$'s cluster. So $x\bigl(B(p,2)\bigr) = 1$ for every~$p$.

\emph{The optimal diversity is $\ell^{*} = 1$.} Every distance is at least
$1$, and two points on opposite shores of a cluster are at distance exactly
$1$, so it is enough to show that every exactly fair $S$ intersects both
shores of some cluster. Suppose for the sake of contradiction that it
does not.
Since $X_1$ and $X_4$ hold the two shores of $C_2$, if $S$ took its point of
each group inside $C_2$ it would intersect both shores of $C_2$. So at least
one of the two points lies in $C_1$, and the points of $C_1$ in these groups,
$b_1$ and $b'_1$, form one shore of $C_1$. The same argument applied to $X_2$,
$X_3$ and $C_3$ puts a point of the other shore of $C_1$ in $S$. Then $S$
intersects both shores of $C_1$, a contradiction. So $\dv(S) = 1$ for every
exactly fair $S$, and $\ell^{*} = 1$.
\end{proof}

Together with Theorem~\ref{thm:four}, the lemma shows that the integrality
gap of the ball LP under exact fairness (i.e., the supremum of
$\gamma_{\max}/\ell^{*}$ over instances) is exactly $4$.

The graph and the groups are due to Szab\'{o} and Tardos~\cite{SzaboT06}, who
defined them in the context of studying independent transversals.
Lemma~\ref{lem:witness} adds a metric, chosen so that the ball LP is feasible
at~$\gamma = 4$.

\subsection{Iterative augmentation}\label{sec:proc}

Our algorithm adapts the augmenting procedure behind Haxell's
theorem~\cite{Haxell95a,Haxell01}, as sketched by Graf and
Haxell~\cite{GrafH20}. It builds an independent set one point at a time,
rearranging earlier choices when no new point can be added outright.
The crux is to show that the procedure never runs out of candidates to add.
Their argument for this is purely combinatorial, whereas ours uses the
properties of a fractional solution to the ball LP and the geometry of the
metric. Asadpour et al.~\cite{AsadpourFS12} use a similar idea for the Santa
Claus problem, but with a configuration LP instead of the ball LP\@.

For the rest of the section, all neighborhoods are taken with respect to
$H_{\gamma/4}$, and \emph{independent} means independent in $H_{\gamma/4}$.
Write~$\gr(v)$ for the index of the group containing~$v$, which is well
defined because the groups partition $V$,
so that $v \in X_{\gr(v)}$, and we write $N_M(v) = N(v) \cap M$ for any set
$M \subseteq V$. The algorithm only uses two properties of $x$:
\begin{equation}\label{eq:roundhyp}
  x(X_i) \ge k_i\,, \quad \forall 1 \le i \le m\,, \quad \text{and} \quad
  x\bigl(N^2[v]\bigr) \le 1\,,  \quad \forall v \in V\,.
\end{equation}
The first property is the fairness constraints of $\mathcal{P}_\gamma$. For
the second, $N[u] \subseteq B(u,\gamma/4)$ for every $u$, so
$N^2[v] \subseteq B(v,\gamma/2)$ by the triangle inequality, and the ball
constraint at $v$ gives
$x\bigl(N^2[v]\bigr) \le x\bigl(B(v,\gamma/2)\bigr) \le 1$.

A \emph{partial solution} is an independent set $M \subseteq V$, with
$|M \cap X_i| \le k_i$ for every $1 \le i \le m$. A group~$X_i$ is
\emph{deficient} if $|M \cap X_i| < k_i$. For $v \notin M$ the \emph{blockers}
of $v$ are the vertices of~$N_M(v)$, i.e., the points of~$M$ that keep
$M \cup \{v\}$ from being independent. The procedure runs in \emph{phases}. A
phase begins with a partial solution,~$M$, and the index~$r$ of a chosen
deficient group, its \emph{root}, and it ends as soon as~$|M|$ has grown by
one.

Within a phase, the state is specified by a partial solution~$M$, a sequence of
\emph{centers} $F = \langle z_1,\dots,z_{|F|} \rangle$, and the index,~$r$, of a
root group. We write $F_l = \langle z_1,\dots,z_l \rangle$ for the sequence of
the first~$l$ centers, so $F_0 = \langle\, \rangle$ is the empty sequence and
$F_{|F|} = F$.
We associate with a state~$(M, F, r)$ the set of group indices
\[
  Q = \{r\} \cup \bigl\{ \gr(y) : y \in N_M(z),\ z \in F \bigr\}\,,
\]
of the root group and of every group that contributes a blocker. Note that~$Q$
is defined with respect to a state $(M, F, r)$, which we do not make explicit
as it will always be clear from the context. We further write
$N[F] = \bigcup_{z \in F} N[z]$ and $N^2[F] = \bigcup_{z \in F} N^2[z]$ for
the sets of vertices within one and two steps of a center,
$N_M(F) = \bigcup_{z \in F} N_M(z)$ for the blockers of the centers,
$X(Q) = \bigcup_{i \in Q} X_i$ for the union of the groups indexed by~$Q$, and
$k(Q) = \sum_{i \in Q} k_i$ for their total quota.

\begin{definition}\label{def:candidate}
Given a state~$(M, F, r)$, we call a vertex~$z \notin M$ a candidate for~$F$
if $x(z) > 0$, $\gr(z) \in Q$, and $z \notin N^2[F]$.
\end{definition}

\begin{definition}
We call a state $(M, F, r)$ \emph{good} if it satisfies the following three
conditions:
\begin{description}
\item[(I1)] $M$ is a partial solution.
\item[(I2)] For every $1 \le l \le |F|$, the center $z_l$ is a candidate for
  $F_{l-1}$ and has at least one blocker.
\item[(I3)] The root group $X_r$ is deficient.
\end{description}
\end{definition}

A phase starts at the state $(M, \langle\, \rangle, r)$, where $M$ is the
partial solution the phase inherits and $X_r$ is a deficient group, so (I1)
and (I3) hold while (I2) is vacuous. Lemma~\ref{lem:moveok} shows that the
state remains good throughout the phase, and that the phase ends by adding a
vertex to~$M$. We first derive three properties of good states.

\begin{lemma}\label{lem:derived}
Let~$(M, F, r)$ be a good state. Then:
\begin{description}
\item[(D1)] the blocker sets $N_M(z_1),\dots,N_M(z_{|F|})$ are pairwise
  disjoint, and no $z_l$ lies in $M$;
\item[(D2)] $N_M(F) \subseteq M \cap X(Q)$;
\item[(D3)] $|F| \le |N_M(F)| \le |M \cap X(Q)| \le k(Q) - 1$.
\end{description}
\end{lemma}

\begin{proof}
\textbf{(D1) holds:} By (I2), each $z_l$ is a candidate for $F_{l-1}$. Hence
$z_l \notin M$, and $z_l \notin N^2[F_{l-1}]$, which means that
$N[z_l] \cap N[F_{l-1}] = \varnothing$. For any $j < l$ we have
$N_M(z_j) \subseteq N[z_j] \subseteq N[F_{l-1}]$, while
$N_M(z_l) \subseteq N[z_l]$, so $N_M(z_j)$ and $N_M(z_l)$ are disjoint.

\textbf{(D2) holds:} A blocker $y \in N_M(z_l)$ lies in $M$, and
$\gr(y) \in Q$ by the definition of $Q$, so $y$ lies in $X(Q)$.

\textbf{(D3) holds:} The $|F|$ blocker sets are nonempty by (I2) and disjoint
by (D1), so $|F| \le |N_M(F)|$, and $|N_M(F)| \le |M \cap X(Q)|$ by (D2).
The groups indexed by $Q$ are pairwise disjoint, so
$|M \cap X(Q)| = \sum_{i \in Q} |M \cap X_i|$. Each term of that sum is at
most $k_i$ by (I1), and the term of the root is at most its quota less one by
(I3), so the sum is at most $k(Q) - 1$.
\end{proof}

Everything is in place to describe one iteration of the algorithm, which we
call a \emph{move}. Starting from a good state $(M, F, r)$, a move either
finds a new element to add to the partial solution $M$ (thus ending the phase)
or it takes us to another good state $(M', F', r)$. Algorithm~\ref{alg:move}
gives the details.

\begin{algorithm}[ht]
\caption{One move from a good state,~$(M, F, r)$.}\label{alg:move}
\begin{algorithmic}[1]
\State Choose a candidate~$z$ to add to~$F$ minimizing $|N_M(z)|$
  \label{ln:cand}
\If{$N_M(z) \neq \varnothing$} \Comment{grow}
  \State Append $z$ to $F$
\ElsIf{$X_{\gr(z)}$ is not deficient} \Comment{swap}
  \State Let $z_j \in F$ be a center with a blocker $y \in N_M(z_j)$
    such that $\gr(y) = \gr(z)$ \label{ln:blocker}
  \State $M \gets (M \setminus \{y\}) \cup \{z\}$
  \If{$N_M(z_j) \neq \varnothing$}
    \State $F \gets F_j$
  \Else
    \State $F \gets F_{j-1}$
  \EndIf
\Else \Comment{the phase ends}
  \State $M \gets M \cup \{z\}$
\EndIf
\end{algorithmic}
\end{algorithm}

\subsection{Correctness}

The correctness proof of the procedure comprises two lemmas.
Lemma~\ref{lem:welldef} shows that every line of Algorithm~\ref{alg:move} is
defined for a good state, so the procedure is never stuck.
Lemma~\ref{lem:moveok} shows that a move either ends the phase or takes a good
state to another good state.

\begin{lemma}\label{lem:welldef}
Let $(M, F, r)$ be a good state. Then every line of Algorithm~\ref{alg:move}
is well defined: The state always has a candidate $z$ for Line~\ref{ln:cand}
to choose from, and if Line~\ref{ln:blocker} is reached, then there exists a
center $z_j$ having a blocker $y$ with $\gr(y) = \gr(z)$.
\end{lemma}

\begin{proof}
We first produce a candidate.
Let~$T$
be~$(M \cap X(Q)) \setminus N^2[F]$, and let~$t = |T|$.
The sets~$N_M(F)$ and~$T$ are disjoint subsets of $M \cap X(Q)$: the first
lies in $M \cap X(Q)$ by~(D2), and in $N^2[F]$, which $T$ avoids. Hence
$|N_M(F)| + t \le |M \cap X(Q)| \le k(Q) - 1$ by~(D3). Now, we bound the
$x$-value of the second (two-hop) neighborhood of~$F$,
\begin{equation}
  x\bigl(N^2[F]\bigr)
  \le \sum_{l=1}^{|F|} x\bigl(N^2[z_l]\bigr)
  \le |F|\,,\label{eqn:n2f}
\end{equation}
where we apply the second property of~\eqref{eq:roundhyp} to each set
$x(N^2[z_l])$. Each of the $t$ vertices of $T$ has $x$-value at most~$1$,
so with $|F| \le |N_M(F)|$ from~(D3), we have via~\eqref{eqn:n2f}
\begin{equation}
  x\bigl(N^2[F] \cup (M \cap X(Q))\bigr)
  \le x\bigl(N^2[F]\bigr) + t
  \le |F| + t \le |N_M(F)| + t \le k(Q) - 1\,.
  \label{eqn:kq1}
\end{equation}
On the other hand, groups indexed by~$Q$ are pairwise disjoint, so the
first property of~\eqref{eq:roundhyp} gives
\begin{equation}
  x\bigl(X(Q)\bigr) = \sum_{i \in Q} x(X_i)
  \ge \sum_{i \in Q} k_i = k(Q)\,.
  \label{eqn:kq}
\end{equation}
Every vertex of $X(Q) \cap \bigl(N^2[F] \cup M\bigr)$ lies in
$N^2[F] \cup (M \cap X(Q))$, and $x \ge 0$, so
\[
  x\bigl(X(Q) \setminus (N^2[F] \cup M)\bigr)
  \ge x\bigl(X(Q)\bigr) - x\bigl(N^2[F] \cup (M \cap X(Q))\bigr)
  \ge k(Q) - \bigl(k(Q) - 1\bigr) = 1\,,
\]
where we apply inequalities~\eqref{eqn:kq} and~\eqref{eqn:kq1} at the end.
Hence there exists some $z \in X(Q)$ with $z \notin N^2[F]$ and
$z \notin M$, but with
 $x(z) > 0$. Since $\gr(z) \in Q$, by the definition of $X(Q)$, so $z$ is a
candidate for~$F$ by
Definition~\ref{def:candidate}.

Now suppose Line~\ref{ln:blocker} is reached. There $X_{\gr(z)}$ is not
deficient, while $X_r$ is deficient by (I3), so $\gr(z) \neq r$. The vertex
$z$ is a candidate, so $\gr(z) \in Q$. Every index of $Q$ other than $r$ is
the group of a blocker of some center, so there are a center $z_l \in F$ and a
blocker $y \in N_M(z_l)$ with $\gr(y) = \gr(z)$. That is the pair needed in
Line~\ref{ln:blocker}.
\end{proof}

\begin{lemma}\label{lem:moveok}
Let $(M, F, r)$ be a good state. Then a move either ends the phase, turning
$M$ into a partial solution with one more vertex,
or takes the state to another good state,~$(M', F', r)$. Moreover, if the move
swaps at the center $z_j$ chosen in Line~\ref{ln:blocker}, then
$N_{M'}(z_l) = N_M(z_l)$ for every $1 \le l < j$ and
$|N_{M'}(z_j)| = |N_M(z_j)| - 1$; and if it drops $z_j$, then $z_j$ is a
candidate for $F_{j-1}$ in the resulting state.
\end{lemma}

\begin{proof}
Let $z$ be the candidate chosen in Line~\ref{ln:cand}.
Algorithm~\ref{alg:move} then takes one of three branches, according to
whether $N_M(z)$ is empty and, if it is, whether $X_{\gr(z)}$ is deficient.
Below we consider each case in turn.

\emph{Grow.} Here $N_M(z) \neq \varnothing$.
The set $M$ is untouched, so (I1) and (I3) still hold. Appending $z$ preserves
(I2): Line~\ref{ln:cand} chose $z$ as a candidate for $F$, which is what (I2)
asks of the new center $z_{|F|+1} = z$, and $z$ has a blocker.

\emph{Swap.} Here $N_M(z) = \varnothing$ and $X_{\gr(z)}$ is not deficient.
Let $z_j$ and $y$ be as in Line~\ref{ln:blocker}, which
Lemma~\ref{lem:welldef} shows exist. Put
$M' = (M \setminus \{y\}) \cup \{z\}$. It is independent, because $M$ is
independent and $N_M(z) = \varnothing$. From $y \in M$, $z \notin M$ and
$\gr(y) = \gr(z)$ we get $|M' \cap X_i| = |M \cap X_i|$ for every $i$. So (I1)
holds, and (I3) holds because no group count changes.

Now consider $F_j$ with the new set $M'$. No blocker set gains $z$, since
$z \notin N^2[F]$, and by (D1) the only one to lose $y$ is $N_M(z_j)$. Hence
$N_{M'}(z_l) = N_M(z_l)$ for $1 \le l < j$, while
$N_{M'}(z_j) = N_M(z_j) \setminus \{y\}$, of size one less, as the lemma
claims. So each $z_l$ with $1 \le l \le j$ is still a candidate for $F_{l-1}$:
the set $Q$ is unchanged, being determined by $r$ and the blockers of
$z_1,\dots,z_{l-1}$, all of index below $j$; the condition
$z_l \notin N^2[F_{l-1}]$ is untouched, since neither $z_l$ nor $F_{l-1}$
changed; and $z_l \notin M'$, since $z_l \notin M$ by (D1) and $z_l \ne z$,
because $z_l \in N^2[F]$ while $z \notin N^2[F]$. Taking $l = j$ gives the
last claim of the lemma.

If $N_{M'}(z_j) \neq \varnothing$, then $z_j$ has a blocker as well, the move
keeps $F_j$, and the state $(M', F_j, r)$ is good. Otherwise the move drops
$z_j$, and the state $(M', F_{j-1}, r)$ is good too: (I1) and (I3) were
checked above, and (I2) holds because $F_{j-1}$ carries just the centers
$z_1,\dots,z_{j-1}$, for which both conditions were verified.

\emph{The phase ends.} Here $N_M(z) = \varnothing$ and $X_i$ is deficient,
where $i = \gr(z)$. The set $M \cup \{z\}$ is independent, because $M$ is
independent and $N_M(z) = \varnothing$. It intersects $X_i$ in
$|M \cap X_i| + 1 \le k_i$ vertices, because $z \notin M$ and $X_i$ is
deficient, and it intersects every other group as $M$ does. So $M \cup \{z\}$
is a partial solution and $|M|$ grows by exactly one.
\end{proof}

So a phase never gets stuck. It remains to show that a phase
always terminates, which is the subject of the next subsection.

\subsection{Complexity}

To show that a phase terminates we bound the number of moves. Let~$\mu$ denote
the least number of blockers of a candidate for~$F$, which, by
Lemma~\ref{lem:welldef}, is well defined at every good state; it is the
quantity minimized at Line~\ref{ln:cand} of Algorithm~\ref{alg:move}. The
\emph{signature} of a state is the tuple
$\bigl\langle |N_M(z_1)|,\dots,|N_M(z_{|F|})|, \mu, \infty \bigr\rangle$,
where the symbol $\infty$ exceeds every integer, and signatures are compared
lexicographically. Since $\infty$ occurs in the last position only, no
signature is a prefix of another, so every comparison is decided at a position
where both signatures are defined.

By~(I2) the first~$|F|$ entries are positive integers, and by~(D1) they sum to
$|N_M(F)|$; the entry~$\mu$ is a non-negative integer. Since $Q$ is a set of
distinct indices, $k(Q) \le k$, so (D3) gives $|N_M(F)| \le k-1$ at every good
state.
If $\mu \ge 1$, the candidate $z$ chosen at Line~\ref{ln:cand} has a blocker,
so the move appends $z$ to $F$.
By Lemma~\ref{lem:moveok} the resulting state is good, and its centers are
$z_1,\dots,z_{|F|}$ and $z$, with pairwise disjoint blocker sets by (D1). So
it has $|N_M(F)| + \mu$ blockers in total, and the bound above applies to it:
$|N_M(F)| + \mu \le k-1$.
On the other hand, if $\mu = 0$, then $|N_M(F)| \le k-1$ by the same bound
applied to the state at hand. So, in both cases, the entries of the signature
other than the last sum to at most $k-1$, and raising $\mu$ by one makes them
all positive, with sum~$s$ in the range $1 \le s \le k$. Therefore, they form a
composition of $s$ into positive parts. Hence the number of signatures is at
most
\[
  \sum_{s=1}^{k} 2^{s-1} = 2^{k} - 1\,.
\]

\begin{lemma}\label{lem:descent}
If a move from a good state does not end its phase, the signature after it
is lexicographically smaller than the signature before it.
\end{lemma}

\begin{proof}
Write $a_l = |N_M(z_l)|$ for the blocker counts before the move, and let $z$
be the candidate chosen in Line~\ref{ln:cand}, so that $|N_M(z)| = \mu$. Since
the move does not end the phase, it is a grow or a swap.

\emph{Grow.} The move appends $z$ to $F$ and leaves $M$ untouched, so the
first $|F|$ entries are unchanged and the entry after them is
$|N_M(z)| = \mu$, which is the value the signature already carried in that
position. The two signatures therefore agree in their first $|F|+1$ positions,
while position $|F|+2$ carries $\infty$ before the move and the new value of
$\mu$ after it. So the signature drops.

\emph{Swap.} Let $z_j$ be the center at which the swap happens, chosen in
Line~\ref{ln:blocker}. Position $j$ carries $a_j$ before the move. By
Lemma~\ref{lem:moveok} the swap takes $M$ to a partial solution $M'$ with
$N_{M'}(z_l) = N_M(z_l)$ for every $l < j$ and $|N_{M'}(z_j)| = a_j - 1$. If
$N_{M'}(z_j)$ is nonempty, the move keeps $F_j$ and position $j$ carries
$a_j - 1$ afterwards. If it is empty, the move drops $z_j$, which is then a
candidate for $F_{j-1}$ with no blockers, so the new value of $\mu$ is $0$ and
position $j$ carries it. Either way the two signatures agree in the first
$j-1$ positions and the new value at position $j$ is smaller, so the signature
drops.
\end{proof}

If Line~\ref{ln:cand} did not minimize the number of blockers, the signature
could rise, and a phase could then return to a state it had already visited.
Asadpour et al.~\cite{AsadpourFS12} measure progress the same way, by the
blocker counts of the centers with a final sufficiently large value. They do
not need the entry~$\mu$, and they choose their candidate arbitrarily. When a
swap empties a center's blocker set, they repeat the swap at that center
within the same step, so no state they compare carries a blocker-free
candidate. Algorithm~\ref{alg:move} performs a single swap and returns, and
$\mu$ records the candidate left waiting. Recursing at that center, as they
do, would dispense with $\mu$, at the cost of succinctness in presentation.
Line~\ref{ln:cand} instead takes one with the fewest blockers, as Haxell does
in the proof of her theorem~\cite{Haxell95a}.

\begin{lemma}\label{lem:terminate}
Every phase ends after at most $2^{k}$ moves. At the end of a phase, $M$ is a
partial solution and $|M|$ is one larger than at the start of the phase.
\end{lemma}

\begin{proof}
By Lemma~\ref{lem:welldef} a move is always available at a good state, and by
Lemma~\ref{lem:moveok} a move that does not end the phase leaves a good state,
so the phase never gets stuck. By Lemma~\ref{lem:descent} the signatures
observed at the starts of successive moves of the phase are strictly
decreasing, hence pairwise distinct, and there are fewer than $2^{k}$ of them.
So the phase performs at most $2^{k}$ moves, and it therefore ends.
Algorithm~\ref{alg:move} leaves~$M$ untouched in a grow step and exchanges one
vertex of~$M$ for another in a swap step, so every move other than the last
leaves~$|M|$ unchanged. The last move adds one vertex to~$M$, which remains a
partial solution by Lemma~\ref{lem:moveok}.
\end{proof}

\subsection{Proof of Lemma~\ref{lem:round}}

\begin{proof}
Start with $M = \varnothing$, which is a partial solution, and repeat: If some
group is deficient, let $r$ be the index of such a group and run a phase. By
Lemma~\ref{lem:terminate} each phase ends with a partial solution one vertex
larger, and every partial solution satisfies
$|M| = \sum_i |M \cap X_i| \le \sum_i k_i = k$, so the loop runs at most $k$
phases. When it stops, no group is deficient, so $|M \cap X_i| \ge k_i$ for
every $i$, while $|M \cap X_i| \le k_i$ because $M$ is a partial solution.
Hence $M$ is an independent set of $H_{\gamma/4}$ with $|M \cap X_i| = k_i$
for every $i$, which proves the existence statement.

For the count, there are at most $k$ phases and at most $2^{k}$ moves in each,
by Lemma~\ref{lem:terminate}, hence at most $k\,2^{k}$ moves in total. One
move searches for a candidate with the fewest blockers and then either appends
it or performs one swap. The search scans the groups of $Q$, the set $N^2[F]$
and the set $M$, and the swap scans the blockers to find $y$ and $j$. All of
this is polynomial in $n$ and $k$. Finally $k \le n$, since
$k \le \sum_i x(X_i) = x(V) \le n$ by the first property
of~\eqref{eq:roundhyp} and $x \in [0,1]^V$. The total time is therefore
$n^{O(1)} 2^{O(k)}$.
\end{proof}

\section{Concluding remarks}

Theorem~\ref{thm:four} achieves exact fairness at factor~$4$ in time
$n^{O(1)}2^{O(k)}$, which is polynomial only when $k = O(\log n)$. The
exponential factor in~$k$ arises from bounding the number of moves in a phase
by the number of signatures a state can carry, which is $2^{O(k)}$. The same
obstacle stands in the way of computing the independent transversal that
Haxell's theorem guarantees. The lazy updates of Annamalai~\cite{Annamalai18}
remove it, at the cost of a slightly stronger assumption on the size of the
groups relative to the maximum degree, as shown by Graf et
al.~\cite{GrafH20,GrafHH22}. Whether some constant factor (independent of $m$)
is attainable with exact fairness in polynomial time is open. The best known
algorithm remains the $(m+1)$-approximation of Addanki et
al.~\cite{Addanki0MM22}.

If we relax the fairness requirement, the factor~$2$ of Theorem~\ref{thm:two}
is best possible, since, already at $m = 1$, Addanki et al.~\cite{Addanki0MM22}
rule out a smaller factor unless $\mathrm{P} = \mathrm{NP}$, however far the
requirement is relaxed. If we insist on exact fairness,
Lemma~\ref{lem:witness} shows that no rounding of the ball LP achieves a
factor below~$4$. Whether a stronger LP can attain factor~$2$ with exact
fairness is open.

\bibliographystyle{plainurl}
\bibliography{references}

\end{document}